\documentclass[letterpaper, 10 pt, conference]{ieeeconf} 

\IEEEoverridecommandlockouts                              
\usepackage{hyperref}
\usepackage{cite}
\usepackage{amsmath,amssymb,amsfonts,amsbsy}
\usepackage{algorithm}
\usepackage{algorithmic}
\usepackage{graphicx}
\usepackage{textcomp}
\usepackage{xcolor}
\usepackage{array}
\usepackage{float}
\usepackage{multirow}
\usepackage{times}
\usepackage{comment}
\usepackage{url,theorem}
\usepackage{comment}
\usepackage{svg}
\usepackage{longtable}
\usepackage{multirow}
\usepackage{siunitx}
\usepackage{tabularx} 
\usepackage{bm}
\usepackage{svg}
\usepackage{colortbl} 
\usepackage{dsfont}

\definecolor{lightgray}{gray}{0.9}  

{\theorembodyfont{\slshape}\newtheorem{theorem}{Theorem}[section]}
{\theorembodyfont{\slshape}}
{\theorembodyfont{\slshape}\newtheorem{lemma}[theorem]{Lemma}}
{\theorembodyfont{\slshape}}
{\theorembodyfont{\slshape}}
{\theorembodyfont{\upshape}\newtheorem{definition}[theorem]{Definition}}
{\theorembodyfont{\upshape}}
{\theorembodyfont{\upshape}}
{\theorembodyfont{\upshape}\newtheorem{assumption}[theorem]{Assumption}}
{\theorembodyfont{\upshape}\newtheorem{example}[theorem]{Example}}

\newcommand{\setR}{\mathbb{R}}
\newcommand{\ones}{\mathds{1}}
\newcommand{\zeros}{\vec{0}}
\newcommand{\setX}{\mathbb{X}}
\newcommand{\nonnegR}{\mathbb{R}_{\geq0}}
\newcommand{\posR}{\mathbb{R}_{>0}}
\newcommand{\setRn}{\mathbb{R}^{n}}

\newcommand{\setRnx}{\mathbb{R}^{n_x}}
\newcommand{\setRny}{\mathbb{R}^{n_y}}

\newcommand{\graphA}{\mathcal{A}}
\newcommand{\graphAbar}{\bar{\mathcal{A}}}

\newcommand{\graphC}{\mathcal{C}}

\newcommand{\graphD}{\mathcal{D}}
\newcommand{\graphE}{\mathcal{E}}

\newcommand{\graphG}{\mathcal{G}}
\newcommand{\graphI}{\mathcal{I}}

\newcommand{\graphK}{\mathcal{K}}

\newcommand{\graphM}{\mathcal{M}}
\newcommand{\graphMbar}{\bar{\mathcal{M}}}

\newcommand{\graphP}{\mathcal{P}}

\newcommand{\graphS}{\mathcal{S}}

\newcommand{\graphV}{\mathcal{V}}

\newcommand{\graphX}{\mathcal{X}}

\newcommand{\qo}{q_{\mathrm{e}}}
\newcommand{\ho}{h_{\mathrm{e}}}

\newcommand{\Aadj}{A_{\mathrm{adj}}}

\newcommand{\Ainc}{A_{\mathrm{inc}}}
\newcommand{\graphAwdn}{\mathcal{A}_{\mathrm{WDN}}}
\newcommand{\graphAwdnbar}{\bar{\mathcal{A}}_{\mathrm{WDN}}}

\newcommand{\OpPoint}{(q_{\mathrm{e}}^\top,h_{\mathrm{e}}^\top)^\top}
\newcommand{\strucset}{\{0, \ast, ?\}^{p \times q}}

\newcommand{\MatrinStruc}[1]{#1\!\in\! \mathcal{P}(\mathcal{#1})}
\newcommand{\matM}{\begin{bmatrix}
        \graphA^\top &\graphC^\top
    \end{bmatrix}}
\newcommand{\matMbar}{\begin{bmatrix}
        \bar{\mathcal{A}}^\top &\graphC^\top
    \end{bmatrix}}

\newcommand{\xe}{x_{\mathrm{e}}}

\title{\LARGE \bf
Exploiting structural observability and graph colorability for optimal sensor placement in water distribution networks}

\author{J.J.H. van Gemert, V. Breschi, D.R. Yntema, K.J. Keesman, M. Lazar
\thanks{This work was performed in the corporation framework of Wetsus, European Centre of Excellence for Sustainable Water Technology (www.wetsus.nl). Wetsus is co-funded by the Dutch Government (Ministry of Economic Affairs and Climate Policy, Ministry of Education, Culture and Science and Ministry of Infrastructure and Water Management) and the Province of Fryslan. The authors like to thank the participants of the research theme `Smart Water Grids' for the fruitful discussions and their financial support.}
\thanks{J.J.H. van Gemert, V. Breschi and M. Lazar are with the Department of Electrical Engineering, Eindhoven University of Technology, The Netherlands. E-mail of corresponding author: {\tt\small J.J.H.v.Gemert@tue.nl}}
\thanks{D. Yntema and K.J. Keesman are with Wetsus, Centre of Excellence for Sustainable Water Technology, 8900 Leeuwarden, The Netherlands. }
\thanks{
K.J. Keesman is also with Mathematical and Statistical Methods – Biometris, Wageningen University, Wageningen, The Netherlands.} 
}

\begin{document}
\maketitle
\thispagestyle{empty}
\pagestyle{empty}
\begin{abstract}
    Water distribution networks (WDNs) are critical systems for our society and detecting leakages is important for minimizing losses and water waste. This makes optimal sensor placement for leakage detection very relevant. Existing sensor placement methods rely on simulation-based scenarios, often lacking structure and generalizability, or depend on the knowledge of specific parameters of the WDN as well as on initial sensor data for linearization and demand estimation. 
    Motivated by this, this paper investigates the observability of an entire WDN, based on structural observability theory. This allows us to establish the conditions for the observability of the WDN model, independently of parameter uncertainties. Additionally, a sensor placement algorithm is proposed that leverages such observability conditions and graph theory and accounts for the industrial and material costs. To demonstrate the effectiveness of our approach, we apply it to a hydraulic-transient WDN model.
\end{abstract}

\section{Introduction}\label{Sec: Introduction}
Water distribution networks (WDNs) are essential pressurized systems that provide drinking water to various regions. Undetected leakages and bursts lead to significant water losses and contamination, which can be costly and a waste of precious resources, especially considering the growing issue of water scarcity due to climate change. Even in highly developed countries like the Netherlands, up to $5\%$ of drinking water is lost every year \cite{waternews2023}. In this context, placing sensors to monitor the network and effectively identifying and localizing leaks is essential \cite{annaswamy2024control}. However, placing sensors in a WDN is both costly and time-consuming, making optimal sensor placement critical to maximizing detection capabilities while taking into account resource constraints.

Traditional sensor placement methods in WDNs rely on simulation-based approaches where multiple leakage scenarios are simulated with several sensor combinations to find the best solution \cite{sarrate2014sensor,casillas2013optimal,steffelbauer2014sensor}. Although these methods have promising results overall, they are unstructured and have poor generalizability since leakage detection is guaranteed only across the (finite) scenarios explored in simulation.
Therefore, incorporating control theory, particularly observability theory, could be advantageous for optimal sensor placement in WDNs to improve leakage detection and monitoring. 
Achieving observability allows for the reconstruction of pressure and flow rates throughout the entire WDN. By continuously comparing the reconstructed pressure and flow rates with expected values, any discrepancies exceeding a predefined threshold can be flagged as a potential leakage. With accurate reconstruction, these leakages can be detected more accurately. Moreover, using observability theory can result in a more structured sensor placement approach that exploits all possible sensor placement combinations.

In the literature, several methods using observability theory for sensor placement have been proposed. In \cite{geelen2021optimal,bopardikar2021randomized}, two approaches utilizing the observability Gramian to guarantee observability and maximize the minimum eigenvalue of the Gramian to enhance detectability are presented. The method introduced in \cite{montanari2020observability} extends this by examining the ratio between the minimum and maximum eigenvalues of the observability Gramian. However, these methods rely on the assumption that the model parameters are known. 
For WDNs, these parameters are often uncertain or unknown, which can result in unobservable outcomes in scenarios where these values vary.
A method that accounts for some parameter variations is proposed in \cite{taha2021revisiting}. They present a sensor placement algorithm that minimizes the Kalman filter state estimation error. Furthermore, the method in \cite{zhang2024functional} accounts for an even more extensive set of uncertainties by developing a new algorithm based on structural functional observability. Their approach provides necessary and sufficient conditions for observability and offers greedy heuristics for selecting a sensor placement configuration from a predetermined set. 
Although these methods are promising, they either focus on linear models, need previous data for linearization, account for a limited set of uncertainties, or rely on algorithms that do not necessarily yield minimal or optimal sensor configurations. However, WDN models are generally nonlinear \cite{chaudhry2014applied} and subject to parametric uncertainties. Moreover, since sensor placement in WDNs is costly, minimizing the number of sensors is desirable. 
Thus, it would be beneficial to develop a method that considers uncertainties in model parameters and guarantees the observability of the WDN, resulting in an optimal sensor placement configuration.

Motivated by this, this paper aims to exploit structural observability principles based on structural controllability theory \cite{mayeda1979strong,mousavi2017structural,jia2020unifying}, which can also be applied to establish local observability of a WDN model. By utilizing structural observability, we do not rely on the assumption that the system parameters are fully known.
Additionally, we present an optimal sensor placement algorithm that exploits the colorability rule from \cite{jia2020unifying} and graph centrality measures, guaranteeing observability for a nonlinear WDN model with uncertain parameters while considering resource constraints. The sensor placement algorithm incorporates a cost-based search algorithm that accounts for the importance of the possible sensor location, the number of connected states, and the industrial cost of sensor placement associated with sensor installation and material costs. We demonstrate the effectiveness of this algorithm with an example on a generally accepted hydraulic-transient WDN model given in \cite{zeng2022elastic}.

The remainder of the paper is organized as follows. In Section \ref{Sec: Preliminaries}, the considered system dynamics, along with observability theory, structural system theory, and graph theory, are introduced. In Section \ref{Sec: Setting and problem formulation}, the setting and the problem formulation are given. 
In Section \ref{Sec: main results}, the main results are described, including the observability conditions and the sensor placement algorithm. 
Two examples using the introduced optimal sensor placement algorithm are shown in Section \ref{Sec: Examples}, followed by the conclusions in Section \ref{Sec: Conclusions}.

\paragraph*{Notation}
Let $\setR$, $\nonnegR$, and $\posR$ denote the field of real numbers, the set of non-negative reals, and the set of positive reals, respectively. A variable $a\in\{0,1\}$ is called a binary variable. 
For a vector $x\in \setR^{n_x}$, $x_i$ denotes the $i$-th element of $x$, $\sqrt{x}=\begin{bmatrix} \sqrt{x_1} & \sqrt{x_2} & \hdots & \sqrt{x_n} \end{bmatrix}^\top$ denotes the element-wise square root of $x$ and $|x|=\begin{bmatrix} |x_1| & |x_2| & \hdots & |x_n| \end{bmatrix}^\top$ the element-wise absolute value of $x$. 
A vector $\zeros \in\setRn$ denotes a vector that contains all zero elements, whereas a vector that contains only ones is denoted by $\ones\in\setRn$. 
For a matrix $A\in \setR^{n_x\times n_x}$, $A^\top$ denotes its transpose and $A^{-1}$ denotes its inverse. A diagonal matrix $A\in\setR^{n_x\times n_x}$ is denoted as $\text{diag}\{A_i\}_{i=1}^{n_x}$. The operator $\bigcup_{j=1}^{n} A_j$ denotes the union of sets $A_1$ through $A_n$, i.e., $\bigcup_{j=1}^{n} A_j = A_1 \cup A_2 \cup \hdots \cup A_n$.

\section{Preliminaries}\label{Sec: Preliminaries}
We now introduce the preliminaries on observability theory, structured systems, and graph theory. Note that the class of systems considered in this section is already a generalization of the WDN model described in Section \ref{Sec: Setting and problem formulation}.
\subsection{Preliminaries on observability theory}

Consider the continuous-time, autonomous, time-invariant, nonlinear system 
\begin{equation}\label{eq: nonsys}
\begin{aligned}
    \Dot{x}(t)=&f(x(t)), \quad t\in \nonnegR,\\
    y(t) =& Cx(t),
    \end{aligned}
\end{equation}
where $x(t)\in \setX\subseteq \setRnx$ and $y(t)\in\setRny$ are the state and the output of the system, respectively,  $f:\setRnx\rightarrow\setRnx$, and $C\in\setR^{n_y\times n_x}$ is the output measurement matrix. In this paper, we will mainly focus on the linearization of \eqref{eq: nonsys} around an arbitrary equilibrium point $\xe\in\setX$ with $f(\xe)=0$, given by
\begin{equation}\label{eq: linsys}
\begin{aligned}
    \delta\Dot{x}(t)=&A(\xe)\delta x(t), \quad t\in \nonnegR,\\
    \delta y(t) =& C\delta x(t),
    \end{aligned}
\end{equation}
where $\delta x(t)=x(t)-\xe$ and  $A(\xe)\in\setR^{n_x\times n_x}$ is the state matrix.
When studying the observability of a system, we aim to determine whether the system's internal state $x(t)$ can be estimated from its available output measurements over time according to the definition below. 
\begin{definition}[\hspace{-0.005cm}\cite{kalman1960general}]\label{Def: observability}
    A system \eqref{eq: nonsys} or \eqref{eq: linsys} is said to be observable if, for any unknown initial state $x(0)\in\setX$, there exists a finite time such that the output $y(\cdot):[0,t]\rightarrow\setRny$ suffices to uniquely determine $x(0)$.
\end{definition}
Given any equilibrium point $\xe\in\setX$, local observability of the nonlinear system \eqref{eq: nonsys} can be studied using the following definition (see \cite[Section 6.1]{lee1967foundations}). 
\begin{definition}\label{Def:Lin-NonlinOBSV}
   If the linearized system \eqref{eq: linsys}, derived from \eqref{eq: nonsys} around an arbitrary equilibrium point $\xe\in\setX$, is observable, then we call the nonlinear system locally observable.
\end{definition}
Next, we introduce structural observability, which analyzes the observability of \eqref{eq: nonsys} through \eqref{eq: linsys} while accounting for uncertainties in the parameters and operating points.

\subsection{Structured linear systems}
Structural system theory uses so-called pattern matrices to evaluate the observability of uncertain linear systems solely based on their structure. To that end, let us first introduce the concept of a pattern matrix $\graphX$.
\begin{definition}\label{Def: PatternMatrix}
Consider a family of matrices $X\in\setR^{p\times q}$ sharing the same structure, defined by the pattern matrix $\graphX\in\strucset$. Then we say that $X$ belongs to the class of pattern matrices $\graphP(\graphX)$ defined as
\begin{align*}
    \graphP(\graphX):=\{X\in\setR^{p\times q}\;|\; &X(i,j)=0\text{ if } \graphX(i,j)=0,\\  &X(i,j)\neq 0 \text{ if } \graphX(i,j)=\ast ,\\ &X(i,j)\in\setR \text{ if } \graphX(i,j)=?\;\}.
\end{align*}
\end{definition}
According to Definition \ref{Def: PatternMatrix}, for system \eqref{eq: linsys}, we define the pattern matrices $\graphA\!\in\!\{0,\ast,?\}^{n_x\times n_x}$ and $\graphC\!\in\!\{0,\ast,?\}^{n_y\times n_x}$, such that $\MatrinStruc{A}$ and $\MatrinStruc{C}$, representing the family of systems \eqref{eq: linsys}. We refer to this family of systems as a \emph{structured system}, denoted by the pair of pattern matrices $(\graphA,\graphC)$. 

The notation of a structured system allows for an extension of the concept of observability introduced in Section \ref{Sec: Preliminaries} to that of structural observability, which is formalized as follows.
\begin{definition}\label{Def: strucobs}
    A structured system $(\graphA,\graphC)$ is called strongly structurally observable if the pair $(A,C)$ is observable for all $\MatrinStruc{A}$ and $\MatrinStruc{C}$.
\end{definition}
In \cite{jia2020unifying}, necessary and sufficient conditions for establishing the strong structural controllability of a structured system are introduced. By exploiting the duality between controllability and observability \cite{kalman1960general}, these conditions can also be used to verify the strong structural observability of the pair $(\graphA,\graphC)$ by applying the theory to $(\graphA^\top, \graphC^\top)$. 
To assess strong structural observability and evaluate the importance of nodes within a network, we use graph-based methods. To facilitate this, we first introduce the necessary preliminaries on graph theory. 
\begin{figure}
    \centering
    \includegraphics[width=0.4\textwidth]{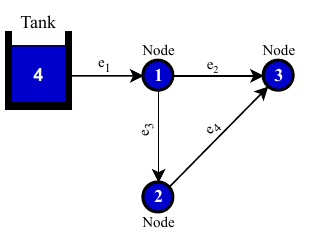}
    \caption{A schematic view of a triangular WDN.}
    \label{fig:WDN network}
\end{figure}

\subsection{Preliminaries on graph theory}
A graph is defined as $\graphG=\{\graphV,\graphE\}$, where $\graphV=\{v_i\}_{i=1}^{n}$ is the set of nodes and $\graphE=\{e_i\}_{i=1}^{m}$ is the set of edges (or links). 
A graph is undirected if each edge connects nodes bidirectionally, and a graph is directed if each edge has a direction, indicating a one-way relation from one node to another.
For both directed and undirected graphs, the adjacency matrix $A_{\mathrm{\mathrm{adj}}}$, indicates whether pairs of nodes are connected by an edge. For an undirected graph, the adjacency matrix is always symmetric, namely $A_{\mathrm{adj}}(i,j) = A_{\mathrm{adj}}(j,i)=1$ if there is an edge between $v_j$ and $v_i$, and $0$ otherwise. Meanwhile, for a directed graph, $A_{\mathrm{adj}}(i,j) = 1$ if there is an edge from $v_j$ to $v_i$, but $A_{\mathrm{adj}}(i,j)=0$ if there is no edge from $v_j$ to $v_i$.
For directed graphs, the incidence matrix $A_{\mathrm{inc}}$ indicates the relationship between nodes and edges, i.e., $\Ainc(i,j) = 1$ if $v_i$ is the tail of edge $e_j$ and $\Ainc(i,j) = -1$ if $v_i$ is the head of edge $e_j$. As an example, consider the WDN in Fig. \ref{fig:WDN network}, the corresponding adjacency and incidence matrices are
\begin{equation}\label{eq:IncAdjWdn}
   \Aadj\!=\!\begin{bmatrix}
         0& 0& 0& 1 \\
         1& 0& 0& 0\\
         1& 1& 0& 0\\
         0& 0& 0& 0\\
    \end{bmatrix}\!\!, \; \Ainc\!=\!\begin{bmatrix}
        -1\!&\! 1\!&\! 1\!&\! 0 \\
         0\!&\! 0\!&\!-1\!&\! 1\\
         0\!&\!-1\!&\! 0\!&\!-1\\
         1\!&\! 0\!&\! 0\!&\! 0\\
    \end{bmatrix}\!\!.
\end{equation}
To assess the importance of each node, centrality measures are used in graph theory to evaluate a node's influence within a network. While various centrality algorithms exist, we will focus on PageRank centrality for sensor placement. Pagerank centrality not only considers the number of connections of a node but also the significance of the nodes it connects to. The PageRank vector containing the ranking of each node denoted by ${\mathrm{PR}}$, is given as the unique positive solution to
 \begin{equation}\label{eq: PageRank}
    {\mathrm{PR}}(i)=\alpha\sum^n_{j=1}\frac{\Aadj(i,j)}{d_j^{\;\mathrm{in}}} {\mathrm{PR}}(j)+\frac{1-\alpha}{n}\!,
 \end{equation}
where $\alpha\in [0,1]$ is the convex combination coefficient, $n$ the number of nodes, $\Ainc$ the incidence matrix and $d_j^{\mathrm{\;in}}$ the in-degree of node $j$, \cite{brin1998anatomy}.
The above-defined graph theory will be used throughout the paper to assess strong structural observability as well as the importance of possible sensor locations.

\section{Setting and problem formulation}\label{Sec: Setting and problem formulation}

Let us consider a nonlinear elastic water column (EWC) model for hydraulic transient analysis of WDNs given by dynamics \eqref{eq: nonsys} with \cite{zeng2022elastic}:
\begin{equation}\label{eq: WDNmodel}
    \begin{aligned}
    x(t)&\!=\!\begin{bmatrix}
    q(t)\\h(t)
\end{bmatrix}\!,\\
f(x(t)) &\!=\!
\begin{bmatrix}
- {L}^{-1}  {R} \operatorname{diag}\{| q(t)|\}  q(t)+ {L}^{-1}  {A_{\mathrm{inc}}}^{\!\!\!\!\top}  h(t) \\
F^{-1}\left( {A_{\mathrm{inc}}}  q(t)- {Q}- {D} \sqrt{ h(t)}\right)
\end{bmatrix}\!\!,
\end{aligned}
\end{equation}
where $F=\text{diag}\{\frac{1}{2}|A_{\mathrm{inc}}|C_l+C_n\}$, $Q\in\setR^{n}$  is the outflow and demand, while the remaining parameters are defined in Table \ref{tab: WDNvariables}.
Note that $R$, $L$, and $D$ are diagonal matrices, with $R$ and $L$ both being positive definite, and all elements of $C_l$ and $C_n$ are positive. 
\begin{table}
    \caption{Parameters of the WDN model \eqref{eq: WDNmodel}.}
    \centering
    \begin{tabular}{|c|c|c|}
    \hline
    Parameter&Symbol&Dimension\\
    \hline
    Volumetric flow rate & $q(t)$ & $\setR^{m}$\\
    Piezometric hear & $h(t)$ & $\setR^{n}$\\
    Hydraulic inductance  & $L$ & $\setR^{m\times m}$ \\
    Hydraulic resistance & $R$ & $\setR^{m\times m}$ \\
    Valve and pressure discharge & $D$ & $\setR^{n\times n}$ \\
    Hydraulic link capacitance & $C_l$ & $\setR^{m}$ \\
    Hydraulic node capacitance & $C_n$ & $\setR^{n}$ \\
    Incidence matrix & $\Ainc$ & $\setR^{n\times m}$\\
    \hline 
    \end{tabular}
    \label{tab: WDNvariables}
\end{table}
\begin{assumption}\label{Ass: nonzero flow/head}
The flow rate in each pipe is always positive as well as the piezometric head in each junction, i.e., $q(t)\in\setR^m_{>0}$ and $h(t)\in\setR^n_{>0}$, for all $t\in\nonnegR$. This implies that, $x(t)\in \setX=\setR^m_{>0}\times\setR^n_{>0}$.
\end{assumption}
Note that WDNs are pressurized distribution systems. Therefore, Assumption \ref{Ass: nonzero flow/head} accurately portrays a real-world scenario \cite{chaudhry2014applied}. 

\subsection{Problem formulation}\label{Subsec: Problem formulation}
Let us assume that $L$, $R$, $D$, $C_l$, $C_n$ and $Q$ in \eqref{eq: nonsys} with \eqref{eq: WDNmodel} are \emph{unknown}, and let us assume we can shape the output measurement matrix, defining it as follows: $C\in\setR^{n_y\times n_x}$ where $C(i,j)\in\{1,0\}$ for all $i=1,\ldots,n_y$ and $j=1,\ldots,n_x$, such that $\sum_{j=1}^{n_x}C(i,j)=1$ for all $i=1,\ldots,n_y$ and $\sum_{i=1}^{n_y}C(i,j)\leq1$ for all $j=1,\ldots,n_x$. Note that these assumptions on the output matrix imply that each row contains exactly one element equal to $1$, and the rows are linearly independent. This specific format is chosen to represent a realistic WDN scenario, as pressure and flow sensors in WDNs cannot be connected to multiple states, e.g., a flow sensor is restricted to measure the flow rate in one specific pipe only.

Under these assumptions, our first objective is to study the observability of \eqref{eq: nonsys} with \eqref{eq: WDNmodel} and to establish the conditions that ensure the observability of the WDN model. 

By relying on the solution to the former problem, our second objective is to develop a sensor placement approach that minimizes $n_y$, i.e., the number of sensors, guaranteeing the observability of the WDN model while considering industrial costs and relying on the structural properties of the WDN model only.

\section{Main results}\label{Sec: main results}
In this section, we address the two main objectives. The first subsection focuses on studying the observability of the nonlinear WDN model, while the second subsection introduces an algorithm aimed at tackling our second objective.

\subsection{Structural observability of the WDN model}\label{Subsec: Observability of WDN}
To investigate the observability of the WDN model in \eqref{eq: nonsys} and \eqref{eq: WDNmodel} without estimating its unknown parameters, we employ a structural approach for linear systems inspired by \cite{jia2020unifying}. 
To this end, we first linearize the WDN model around an equilibrium point $\xe=\OpPoint\in\setX$, leading to \eqref{eq: linsys} with
 \begin{align}\label{eq: linear WDN model}
     A(\xe) \!=\!\! \begin{bmatrix}\!
         -L^{-1}R\text{ diag}\begin{Bmatrix}2|{\qo}|\end{Bmatrix}\!\!&\!\! L^{-1}\Ainc^\top\!\\
         \!F^{-1}\Ainc\!\!&\!\!-F^{-1}D\text{diag}\begin{Bmatrix}\!\frac{1}{2\sqrt{{\ho}}}\!\!\end{Bmatrix}
    \!\end{bmatrix}\!\!. 
\end{align}
Next, we derive a structured representation for the linearized WDN model \eqref{eq: linsys} with \eqref{eq: linear WDN model}. We begin by defining a pattern matrix $\graphC\in\{0,\ast\}^{n_y\times n_x}$ with
\begin{equation}\label{eq: strucC}
    \graphC(i,j)=\begin{cases}
        0 &\text{if } C(i,j)=0 \\ \ast &\text{otherwise},  \end{cases}
\end{equation}
such that $C\in\graphP(\graphC)$. 
We now focus on $A$ in \eqref{eq: linear WDN model}. Since $L$ and $R$ are nonzero diagonal matrices, under Assumption \ref{Ass: nonzero flow/head} its first block is a diagonal matrix with nonzero entries. We can thus define $\graphI=\text{diag}\{\ast\}\in\{0,\ast\}^{m\times m}$ such that 
\begin{equation*}
-L^{-1}R\text{diag}\begin{Bmatrix}2|{\qo}|\end{Bmatrix}\in\graphP(\graphI).
\end{equation*}
For the fourth block, a similar reasoning holds, leading to
\begin{equation*}
-F^{-1}D\text{diag}\begin{Bmatrix}\frac{1}{2\sqrt{{\ho}}}\end{Bmatrix}\in\graphP(\graphD).
\end{equation*}
Nonetheless, due to the arbitrary nature of the discharge matrix $D$, $\graphD=\text{diag}\{?\}\in\{0,?\}^{n\times n}$.

Next, we look at the second and third blocks. The multiplication of the diagonal nonzero matrices $L$ and $F$ with the (transpose) incidence matrix leads to a pattern matrix $\graphA_{\mathrm{inc}}\in\{0,\ast\}^{n\times m}$ with
\begin{equation*}
    \graphA_{\mathrm{inc}}(i,j)=\begin{cases}0 &\text{if } A(i,j)=0 \\\ast &\text{otherwise},\\ \end{cases}
\end{equation*}
such that
\begin{align*}
    L^{-1}\Ainc^\top\in\graphP(\graphA_{\mathrm{inc}}^\top),\quad
        F^{-1}\Ainc\in\graphP(\graphA_{\mathrm{inc}}),
\end{align*}
which represents the coupling between the two sets of states in \eqref{eq: linear WDN model}. We can now define the pattern matrix
\begin{equation}\label{eq: structuredAwdn}
    \graphAwdn=\begin{bmatrix}
        \graphI & \graphA_{\mathrm{inc}}^\top\\
        \graphA_{\mathrm{inc}} & \graphD\\
   \end{bmatrix}\!\!,
\end{equation} 
where $\graphAwdn\in\{0,\ast,?\}^{n_x\times n_x}$ such that $A(\xe)\in\graphP(\graphAwdn)$. We can now make some conclusions on the observability of the nonlinear WDN model based on the above-defined structured system $(\graphAwdn,\graphC)$. 
\begin{lemma}\label{Cor: global observability}
Suppose Assumption \ref{Ass: nonzero flow/head} holds and the structured system $(\graphAwdn,\graphC)$ with \eqref{eq: strucC} and \eqref{eq: structuredAwdn}, is strongly structurally observable. Then, the nonlinear WDN model \eqref{eq: nonsys} with \eqref{eq: WDNmodel} is locally observable.
\end{lemma}
\begin{proof}
This result follows directly from three key points. First, by Definition \ref{Def:Lin-NonlinOBSV}, a nonlinear system is locally observable if its linearization is observable for some $\xe\in\setX$.
Second, Assumption \ref{Ass: nonzero flow/head}, specifies that the states of the WDN model operate within a specific set, namely $q(t)\in\setR^m_{>0}$ and $h(t)\in\setR^n_{>0}$, for all $t\in\nonnegR$. This implies that any equilibrium point $\xe=\OpPoint$ lies within the set $\setR^m_{>0}\times\setR^n_{>0}$ ensuring that $A(\xe)\in\graphP(\graphAwdn)$ for all possible eqiulibrium points. Finally, since the structured system $(\graphAwdn,\graphC)$ is observable, then by Definition \ref{Def: strucobs}, the pair $(A(\xe), C)$ is observable for all $A(\xe)\in\graphP(\graphAwdn)$ and $\MatrinStruc{C}$, thus concluding the proof.
\end{proof}

To verify observability of the structured linearized WDN system, we use the colorability approach from \cite{jia2020unifying}, which uses a color-change rule to assess structural controllability, and by duality, observability.

Before introducing colorability, we need to introduce two additional matrices $\graphM=\matM$ and $\graphMbar=\matMbar$, where 
\begin{equation}
    \graphAbar(i,i):= \begin{cases}\ast & \text { if } \graphA(i,i)=0  \\ ? & \text { otherwise.}\end{cases}
\end{equation}
We can associate a directed graph $\graphG(\graphM)=(V,E)$ as follows.
The set of nodes is taken as $V=\{1,\ldots,n_x,\ldots,n_x+n_y\}$. For the edges we include $(j,i)\in E_\ast\in V\times V$ if $M(i,j)=\ast$ and  $(j,i)\in E_?\in V\times V$ if $M(i,j)=?$. To distinguish between these edges visually, we represent the edges in $E_\ast$ with solid arrows and those in $E_?$ with dashed arrows. 
For a graph $\graphG(\graphM)$, one can define the \textbf{color change rule} as follows. 
Every node in $V$ is originally colored white. If node $j$ has exactly one out-neighbor $i$ with $(j,i)\in E_\ast$ that is white (self-loops included), then we change the color of node $i$, to white. If the nodes $V(i)$ for $i=1,\ldots,n_x$, i.e., all the nodes that represent the states of $\graphA^\top$, are colored, then the graph is called \emph{colorable} \cite{jia2020unifying}.

\begin{example}
Consider the following pattern matrices
    \begin{equation}\label{eq: exampleAC}
\graphA^\top\!\!=\!\!\begin{bmatrix}
            0     & \ast & \ast \\
            \ast  & \ast    & 0 \\
            \ast  & ?    & \ast    
        \end{bmatrix}, \;\graphAbar^\top\!\!=\!\!\begin{bmatrix}
            \ast     & \ast & \ast \\
            \ast  & ?    & 0 \\
            \ast  & ?    & ?   
        \end{bmatrix} \text{ and }\graphC^\top\!\!=\!\!\begin{bmatrix}
            \ast \\
          0\\
           \ast   
        \end{bmatrix}\!\!,
    \end{equation}
    which form the combined matrices
    \begin{equation}\label{eq: example M}
        \graphM\!\!=\!\!\begin{bmatrix}
            0     & \ast & \ast & \ast\\
            \ast  & \ast    & 0 & 0\\
            \ast  & ?    & \ast    & \ast\\
        \end{bmatrix} \text{ and } \graphMbar\!\!=\!\!\begin{bmatrix}
            \ast     & \ast & \ast & \ast\\
            \ast  & ?    & 0 & 0\\
            \ast  & ?    & ?    & \ast\\
        \end{bmatrix}\!\!.
    \end{equation}
The associated graphs $\graphG(\graphM)$ and $\graphG(\graphMbar)$ are shown in the first column with 2 plots 
in Fig. \ref{fig:colorability}. In examining the upper left graph in Fig. \ref{fig:colorability} linked to $\graphG(\graphM)$, we see that each node has two out-neighbors that are white, This means that no node can be colored, rendering the graph non-colorable.
Now, let us adjust $\graphC$ as follows: 
\begin{equation}\label{eq: example M new}
         \graphC^\top\!\!=\!\!\begin{bmatrix}
            \ast  &0\\
          0 & 0\\
           \ast &\ast   
        \end{bmatrix}\!\!,\;
       \graphM\!\!=\!\!\begin{bmatrix}
            0     & \ast & \ast & \ast &0\\
            \ast  & \ast    & 0 & 0&0\\
            \ast  & ?    & \ast    & \ast& \ast\\
        \end{bmatrix}\!\!.
    \end{equation}
    As shown in columns 2 and 3 in Fig. \ref{fig:colorability}, with this new $\graphC$ matrix, the graph $\graphG(\graphM)$ is now colorable, the same holds for $\graphG(\graphMbar)$. \hfill$\square$
\end{example}
\begin{figure}[ht]
     \centering
     \includegraphics[width=0.99\linewidth]{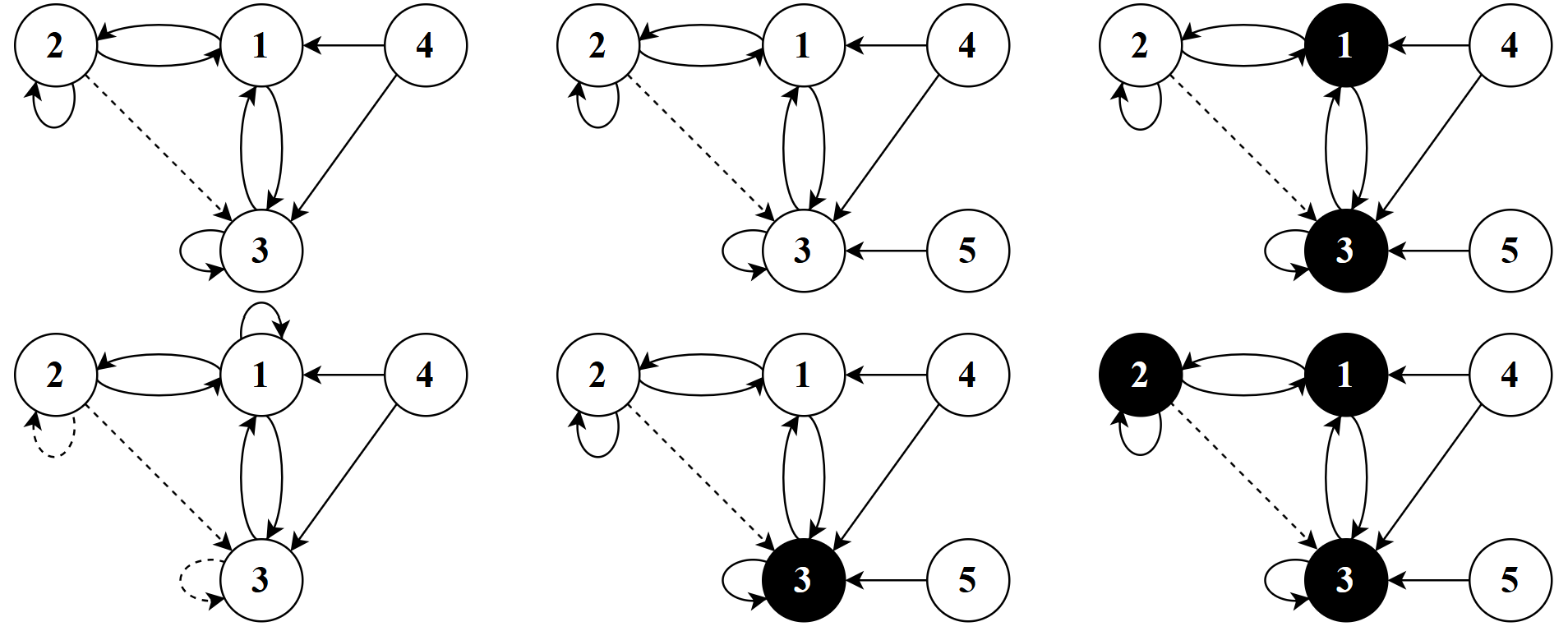}
     \caption{Examples of the color change rule based on pattern matrices $\graphM$ and $\graphMbar$ in \eqref{eq: example M} and $\graphM$ in \eqref{eq: example M new}.}
     \label{fig:colorability}
\end{figure}

Based on this color change rule, we now give the condition for which the structured system is strongly structurally observable.
\begin{definition}[\hspace{-0.005cm}\cite{jia2020unifying}]\label{Def: color-observ}
     A structured system $(\graphA,\graphC)$ is strongly structurally observable if the graphs associated with the pattern matrices $M=\begin{bmatrix}\graphA^\top & \graphC^\top \end{bmatrix}$ and $\graphMbar=\begin{bmatrix}\graphAbar^\top & \graphC^\top \end{bmatrix}$ are colorable. \end{definition}
Based on Definition \ref{Def: color-observ} and Corollary \ref{Cor: global observability}, the WDN model \eqref{eq: nonsys} with \eqref{eq: WDNmodel} is observable if the graphs $\graphM=\begin{bmatrix}\graphAwdn^\top & \graphC^\top \end{bmatrix}$ and $\graphMbar=\begin{bmatrix}\graphAwdnbar^\top & \graphC^\top \end{bmatrix}$ are colorable.  

Next, we define a search algorithm to select the optimal measurement matrix $C$, which results in a colorable combination of $(\graphA^\top,\graphC^\top)$. Note that this implicitly yields a sensor placement solution via the structure of the resulting C matrix. 

\subsection{Sensor placement algorithm}
To set up the sensor placement algorithm, we first introduce a procedure that verifies the colorability of a graph $\graphG(\graphM)$ following the color change rule, detailed in Algorithm \ref{Alg:colorability}.  
The vector \emph{Color} in line 3 is a binary vector where \(\text{Color}\{i\} = 0\) signifies a black node and \(\text{Color}\{i\} = 1\) a white node. Entries in $\graphM$ marked with $?$ are ignored, as they do not affect colorability (line 4).
The algorithm computes each node’s out-degree (line 5) and places nodes with one outgoing edge into a queue (line 6). It then iteratively colors the out-neighbors of these nodes black, removing their incoming edges from $\graphM$ (line 12) and retaining only edges to white nodes. This process continues until the queue is empty. If all nodes of $\graphA$ are colored black, the graph is colorable.

\begin{algorithm}[t]
    \caption{Colorability Check for Observability}
    \begin{algorithmic}[1]\label{Alg:colorability}
        \STATE \textbf{Input:} $\graphM$ or $\graphMbar \in \setR^{n_x \times (n_x+n_y)}$
        \STATE \textbf{Output:} Binary value \textbf{isColorable}
        \STATE $\text{Color} = \zeros \in \setR^{n_x + n_y}$
        \STATE $\graphM(i,j) = 0 \text{ if }\graphM(i,j) = ?$
        \STATE $c_{\text{out}} = \graphM^\top \ones$
        \STATE$ \text{queue} = \{i \mid c_{\text{out}}\{i\} = 1\}$
        \WHILE{$\text{queue} \neq \emptyset$}
            \STATE $\text{Node} = \text{queue}\{1\}, 
 \text{queue}{1} = [\;]$
            \STATE $e_{\text{out}} = \sum _{j=1}^{n_x} \{ \graphM(\text{Node},j) = \ast \}$
            \IF{$\text{length}(n_{\text{out}}) = 1$}
                \STATE $\text{newNode} = \graphM({\text{Node},n_{\text{out}}})$, $\text{Color}(\text{newNode}) = 1$ 
                \STATE $\graphM(:, \text{newNode}) = 0$ for all in\_edge
                \STATE \text{queue} = [\text{queue}; newNode]            
                \ENDIF
        \ENDWHILE
    \end{algorithmic}
\end{algorithm}

Next, we present the sensor placement search algorithm in Algorithm \ref{Alg:sensor_placement}, which aims to minimize the number of sensors, $n_y$, while accounting for industrial costs. The structure of $C$ is defined in Section \ref{Subsec: Problem formulation}, and we use the term "possible sensor location" to refer to the state or node to which a sensor can be assigned. The overall procedure of the Algorithm is illustrated in the corresponding flowchart in Fig. \ref{fig:flowchart}.
The algorithm takes as input the pattern matrices $\graphA$ and $\graphAbar$, the industrial cost vector $c_{\text{ind}}$, the set of possible sensor nodes $nodes$, and the PageRank values for each node. The PageRak values are calculated by inverting the values from the original PageRank equation \eqref{eq: PageRank} to align centrality with cost, i.e., a higher centrality corresponds to greater interest in a sensor location, thus reducing its cost.

 \begin{figure}[ht]
      \centering
      \includegraphics[width=0.99\linewidth]{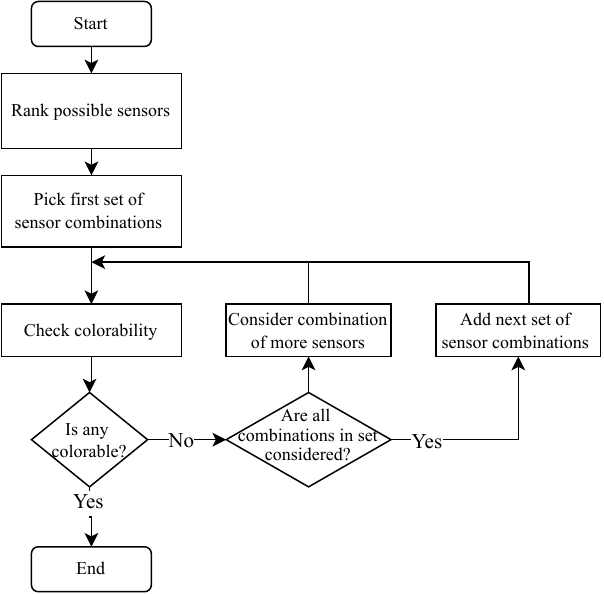}
 \caption{Search algorithm for sensor placement based on colorability.}  
 \label{fig:flowchart}
\end{figure}

The first step in Algorithm~\ref{Alg:sensor_placement} is to compute the in- and out-degree vectors (line 3), $c_{\text{in}}$ and $c_{\text{out}}$, respectively. All cost vectors are assumed to be normalized:
\begin{equation}\label{eq:Norm}
    c_{x_n}=\frac{c_x - \text{min}(c_x)}{\text{max}(c_x) - \text{min}(c_x)},
\end{equation}
such that $c_{x_n}\in[0,1]^n$. 
The normalized cost vectors are weighted by coefficients $\alpha_{\text{PR}}, \alpha_{\text{in}}, \alpha_{\text{out}}, \alpha_{\text{ind}}$, and combined to form the average cost $c_n$ (line 4).
The possible sensor locations are then grouped by average cost from lowest cost to highest cost (line 5).

The algorithm proceeds by iterating over each group, evaluating the possible sensor locations in the $n^{th}$ group alongside those from all preceding groups (i.e., the groups $1$ through $n$, see line 8). All the possible combinations are stored in $\graphK$ (line 10). 
To avoid repetitive evaluations, only combinations that include elements from the current group, $group\{n\}$, are considered.

To slightly enhance computational efficiency, the algorithm iteratively evaluates subsets with possible sensor combinations within $\graphS_n$. 
For each sensor combination, the matrices $\graphC$, $\graphA$ and $\graphAbar$ are defined, and the colorability of $\graphM$ and $\graphMbar$ is evaluated using Algorithm \ref{Alg:colorability} (line 12-14). 
When an observable pair is found, the algorithm terminates and returns $C_{obs}$.
If no observable combination is identified within the current group, the algorithm proceeds to the next group and resumes from line 8. 

\begin{algorithm}[t]
    \caption{Heuristic Search for Minimal Sensor Placement with Colorability Check}
    \begin{algorithmic}[1]\label{Alg:sensor_placement}
    \STATE \textbf{Input:} $\graphA$ and $\graphAbar$ of size $n_x \times n_x$, $c_{\text{PR}},$ and $ c_{\text{ind}}$ of size $n_x$ and $nodes$ of size $n_x$
    \STATE \textbf{Output:} Observable sensor configurations $C_{obs}$
    \STATE $c_{\text{in}}= \graphA \ones$, $c_{\text{out}}= \graphA^\top \ones$
    \STATE $c_n=\alpha_{\text{PR}}c_{\text{PR}}+\alpha_{\text{in}}c_{\text{in}}+\alpha_{\text{out}}c_{\text{out}} +\alpha_{\text{ind}}c_{\text{ind}}$
    \STATE $group = \text{group }nodes \text{ by }c_n$
    \STATE $C_{obs}=[\;]$, $C_{unobs}=[\;]$
    \FOR{$n=1,\ldots,length({group})$}
    \STATE $\graphS_n = \bigcup_{j=1}^{n} {group}\{j\}$
    \FOR{$k=1,\ldots,length(\graphS_n)$}
        \STATE $\graphK\! =\! \{ \text{$k$-combinations in } \graphS_n\! \mid \!(\graphK \!\cap \text{group}\{n\}) \!\geq \!1 \}$
        \FOR{$i=1,\ldots,size(\graphK_n,1)$}
            \STATE \!$\graphC\!=\!\text{diag}(\graphK_n\{:,i\})$, $\graphM \!= \![\graphA^{\!\top}\!,\graphC^{\top}\!]$, $\graphMbar\!=\![\graphAbar^{\top}\!,\graphC^\top\!]$
            \STATE \!Colorability$_1 \leftarrow $ Algorithm \ref{Alg:colorability} with \textbf{input} $\graphM$
            \STATE  \!Colorability$_2 \leftarrow $ Algorithm \ref{Alg:colorability} with \textbf{input} $\graphMbar$ 
            \IF{Colorability$_1 = 1 \;\&\; $Colorability$_2 = 1$}
                \STATE $C_{obs} = [C_{obs};\; \graphC]$
                \ELSE
                \STATE $C_{unobs}= [c_{unobs};\;\graphC]$
            \ENDIF
        \ENDFOR
            \IF{$C_{obs}$ $\neq \emptyset$}
                \STATE $break;$
            \ENDIF
    \ENDFOR
    \ENDFOR
    \end{algorithmic}
\end{algorithm}

It is important to note that this algorithm guarantees an observable solution. In the worst-case scenario, the algorithm will select all possible sensor locations, returning $\graphC$ such that $n_y = n_x$.
Additionally, while it can recover the optimal configuration, the computational complexity grows exponentially with the network size. Given the structure of $C$ in Section \ref{Subsec: Problem formulation}, then with $A \in \setR^{n_x \times n_x}$, there are $2^{n_x} - 1$ possible sensor configurations. As $n_x$ increases, the search space quickly becomes prohibitively large.

\section{Illustrative example}\label{Sec: Examples}
\begin{table}[b]
    \caption{Sensor cost of the WDN network in  \eqref{eq: structuredAwdn}.}
    \centering
    \footnotesize
    \resizebox{0.8\linewidth}{!}{  
    \begin{tabular}{|c|c|c|c|c||>{\columncolor{lightgray}}c||}
    \hline
    \textbf{Sensor location} & $\mathbf{c}_{\textbf{out}}$ & $\mathbf{c}_{\textbf{in}}$ & $\mathbf{c}_{\textbf{PR}}$ & $\mathbf{c}_{\text{ind}}$ & $\mathbf{c}_{\textbf{n}}$ \\
    \hline
     2&         0.75&           0.75&        0.16967&           1&            1\\  
     3&         0.75&           0.75&        0.16967&           1&            1\\  
     5&            1&              1&        0.14991&           0&      0.78423\\  
     4&         0.75&           0.75&        0.14578&      0.4995&      0.78086\\  
     1&         0.75&           0.75&              0&      0.4995&      0.76509\\  
     6&          0.5&            0.5&        0.61626&           0&      0.39658\\  
     7&          0.5&            0.5&        0.61626&           0&      0.39658\\ 
     8&            0&              0&              1&           0&            0\\   
    \hline
    \end{tabular}
    }
    \label{tab: WDNCost}
\end{table}

In this section, we show the results of Algorithm \ref{Alg:sensor_placement} on two examples. 
In the first example, we analyze a star network with an associated pattern matrix:
\begin{equation}\label{eq:Astar}
    \graphA^\top \!\!=\!\! \begin{bmatrix}
     0 & 0 & 0 & \ast & \ast \\
     0 & 0 & 0 & 0 & \ast \\
     0 & 0 & 0 & 0 & \ast \\
     0 & 0 & 0 & 0 & \ast \\
     \ast & \ast & \ast & \ast & 0
    \end{bmatrix}\!.
\end{equation}
The edge between nodes $1$ and $4$, represented by $\graphA^\top(4,1)$ in \eqref{eq:Astar}, introduces a cycle among nodes 1, 5, and 4. This modified star network is chosen to include the presence of cycles typical of most dynamical networks including WDNs. The cost associated with potential sensor locations is listed in Table \ref{tab: StarCost}.
\begin{table}[ht]
    \caption{Sensor costs for the star network in \eqref{eq:Astar}.}
    \centering
    \begin{tabular}{|c|c|c|c|c||>{\columncolor{lightgray}}c||}
    \hline
    \textbf{Sensor location} & $\mathbf{c}_{\textbf{out}}$ & $\mathbf{c}_{\textbf{in}}$ & $\mathbf{c}_{\textbf{PR}}$ & $\mathbf{c}_{\text{ind}}$ & $\mathbf{c}_{\textbf{n}}$ \\
    \hline
    5 & 1 & 1 & 0 & 1 & 1\\
    4 & 0.333 & 0 & 1 & 1 & 0.381 \\ 
    1 & 0 & 0.333 & 0.83503 & 1 &0.000486 \\
    2 & 0 & 0 & 1 & 1 & 0 \\
    3 & 0 & 0 & 1 & 1 & 0 \\
    \hline
    \end{tabular}
    \label{tab: StarCost}
\end{table}
Based on the pattern matrix \eqref{eq:Astar} and the costs from Table \ref{tab: StarCost}, the search algorithm returns the following pattern matrix $\graphC$ as its solution, 
\begin{equation}\label{eq: Cstar}
    \graphC^\top\!=\!\begin{bmatrix}
        \ast& 0& 0\\ 
        0& \ast& 0\\ 
        0& 0& \ast\\
        0& 0& 0\\
        0& 0& 0
    \end{bmatrix}\!, \; C^\top\!=\!\begin{bmatrix}
        1& 0& 0\\ 
        0& 1& 0\\ 
        0& 0& 1\\
        0& 0& 0\\
        0& 0& 0
    \end{bmatrix}\!, 
\end{equation}
resulting in a colorable graph for the structured system $(\graphA, \graphC)$ with \eqref{eq:Astar} and \eqref{eq: Cstar}.
Fig. \ref{fig:star} illustrates the graph $\graphG(\graphM)$ associated with $\graphM=\begin{bmatrix} \graphA^\top & \graphC^\top \end{bmatrix}$. 
\begin{figure}[ht]
     \centering
     \includegraphics[width=0.99\linewidth]{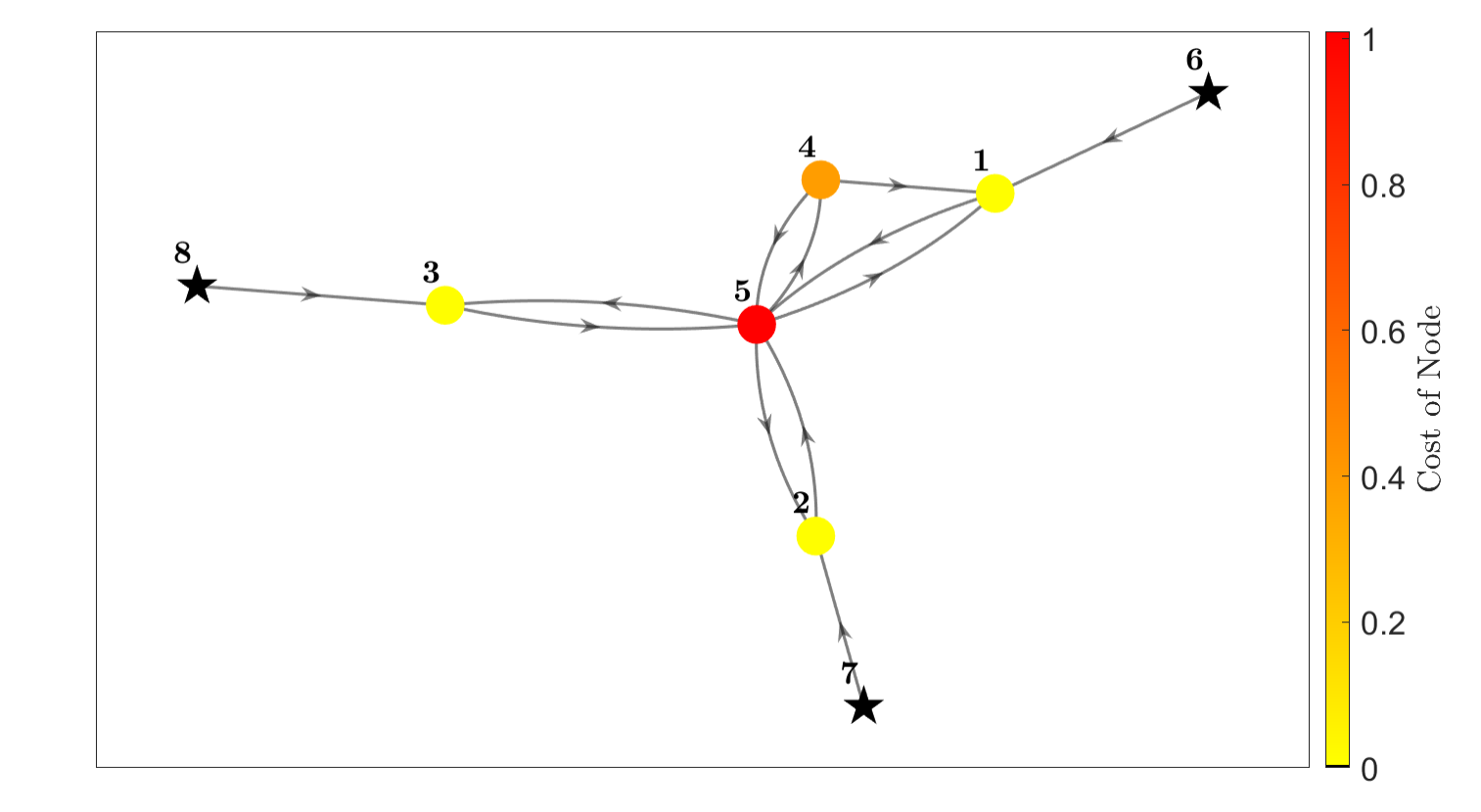}
     \caption{Star network: Colorable graph $\graphG([\graphA^\top,\graphC^\top])$ where a star denotes a sensor, a circle denotes a state and the color denotes the cost of connecting a sensor to that specific state.}
     \label{fig:star}
 \end{figure}
Referring back to the color-change rule outlined in Section \ref{Subsec: Observability of WDN}, we observe that connecting a sensor to nodes $2$, $3$, and $4$ instead of $2$, $3$, and $1$, would also produce a colorable graph. However, due to the significantly lower cost associated with node 1 compared to node 4, the algorithm selects this configuration as the optimal solution.

Next, we show an example using the WDN illustrated in Fig. \ref{fig:WDN network}. The corresponding incidence matrix is given in equation \eqref{eq:IncAdjWdn}, with the pattern matrix $\graphAwdn$ given in equation \eqref{eq: structuredAwdn}. The costs of the possible sensor locations in the WDN are provided in Table \ref{tab: WDNCost}. 

In this network, states $1$ through $4$ represent flow rates in the pipes, while states $5$ through $8$ correspond to the pressure in each junction. Due to the high installation costs of flow sensors, their cost values, represented by $c_{\text{ind}}$ in Table \ref{tab: WDNCost}, are higher. Additionally, we assume that the pipes connected to states 2 and 3 are in remote locations, making their sensor placement undesirable due to the higher installation costs, which further increases the cost values in $c_{\text{ind}}$.

Applying search Algorithm \ref{Alg:sensor_placement} to the pattern matrix $\graphAwdn$ with \eqref{eq: structuredAwdn} with the cost values in Table \ref{tab: WDNCost}, we obtain the following colorable sensor configurations defined in $\graphC_1$ and $\graphC_2$:
\begin{equation}\label{eq: CWdn}
\begin{aligned}
    \graphC_1 &= \begin{bmatrix}
        0 & 0 & 0 & \ast & 0 & 0 & 0 & 0 \\
        0 & 0 & 0 & 0 & 0 & \ast & 0 & 0
    \end{bmatrix}, \\
    \graphC_2 &= \begin{bmatrix}
        0 & 0 & 0 & \ast & 0 & 0 & 0 & 0 \\
        0 & 0 & 0 & 0 & 0 & 0 & 0 & \ast
    \end{bmatrix}
\end{aligned}
\end{equation}

The resulting graphs $\graphM$ in Figure \ref{fig:WDNsens} show that placing a sensor at state 8 (the tank) has a very low cost. However, it would have minimal impact on the colorability of the overall network. Consequently, the search algorithm does not select such a solution. 
Moreover, according to the color-change rule, it is evident that sensor pairs $\{2, 6\}$ or $\{3, 7\}$ would also yield a colorable graph. However, as shown in Table \ref{tab: WDNCost}, these configurations would lead to a higher overall cost, explaining the optimal solution provided by the search algorithm.

 \begin{figure}[t]
\centering
     \includegraphics[width=0.99\linewidth]{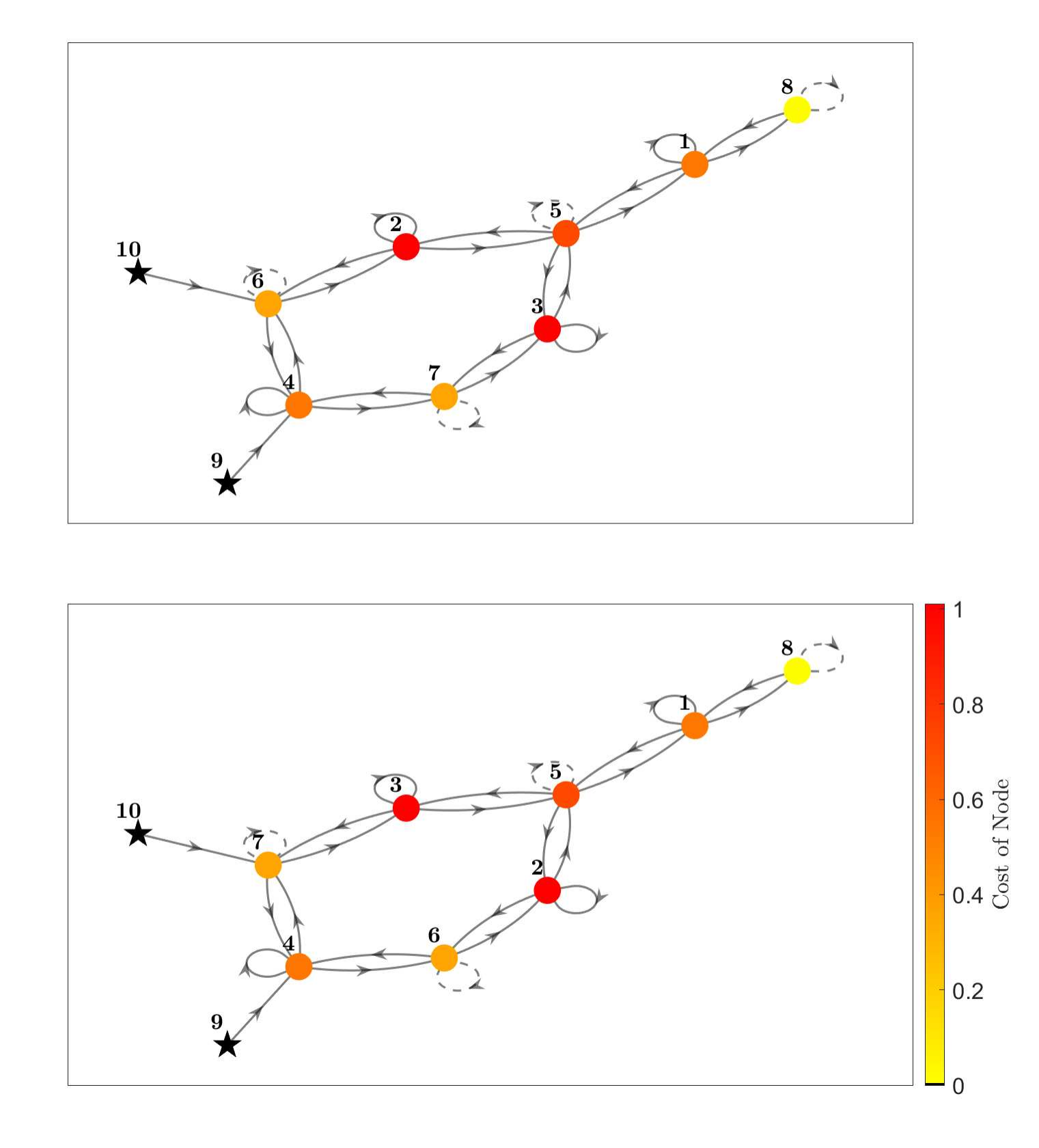}
     \caption{WDN network: Two colorable graphs with different sensor configurations $\graphG([\graphAwdn^\top,\graphC^\top])$, where a star denotes a sensor, a circle denotes a state, and the color denotes the cost of connecting a sensor to that specific state.}
     \label{fig:WDNsens}
 \end{figure}

\section{Conclusions}\label{Sec: Conclusions}
In this paper, we studied the observability of a WDN model using structural observability theory. We established the conditions under which the WDN model is observable, independent of parameter uncertainties, and without the need for initial sensor data. Additionally, we presented a sensor placement algorithm that exploits colorability and graph centrality measures, guaranteeing observability for a nonlinear WDN model with uncertain parameters. The developed search algorithm incorporates a cost-based search that considers the importance of each possible sensor location, the number of connected states, and the industrial costs. The proposed search algorithm was used to construct observable configurations for two examples including a WDN. Through these examples, we demonstrated that the proposed approach allows for cost-effective sensor placement and can be tailored to meet the preferences of the water industry. 

For future work, some directions can be explored to further build on the results of this paper. One potential direction is to extend the proposed sensor placement algorithm to handle larger, more complex WDNs. Additionally, reducing the conservatism of the proposed algorithms, aiming to achieve observability with fewer sensors while maintaining the detectability of leakages.

\paragraph*{Acknowledgements}
The authors are grateful to Dr. Xiaodong Cheng for his useful comments.
\bibliographystyle{IEEEtran}
\bibliography{main}
\end{document}